\documentclass[onecolumn,10pt]{IEEEtran}

\usepackage{amsfonts,amsmath,amssymb,amsthm,caption,algorithmic}
\usepackage{amsbsy}
\usepackage{graphicx,multirow,bm}
\usepackage{color}
\usepackage{subfigure}
\makeatletter
\newtheoremstyle{mythm}{3pt}{3pt}{}{16pt}{\bfseries}{:}{.5em}{}
\theoremstyle{mythm}
\newtheorem{theorem}{Theorem}
\newtheorem{example}{Example}
\newtheorem{definition}{Definition}

\newtheorem{fact}{Fact}

\newtheorem{corollary}{Corollary}
\newtheorem{lemma}{Lemma}
\newtheorem{construction}{Construction}
\begin{document}
\title{The Optimal Placement Delivery Arrays
\author{Minquan Cheng, Jing Jiang, Qifa Yan, Xiaohu Tang, \IEEEmembership{Member,~IEEE} and Haitao Cao
}
\thanks{M. Cheng and J. Jiang are with Guangxi Key Lab of Multi-source Information Mining $\&$ Security, Guangxi Normal University,
Guilin 541004, China (e-mail: $\{$chengqinshi,jjiang2008$\}$@hotmail.com).}
\thanks{Q. Yan and X. Tang are with the Information Security and National Computing Grid Laboratory,
Southwest Jiaotong University, Chengdu, 610031, China (e-mail: qifa@my.swjtu.edu.cn, xhutang@swjtu.edu.cn).}
\thanks{H. Cao is with the Institute of Mathematics, Nanjing Normal University, Nanjing 210023, China (e-mail: caohaitao@njnu.edu.cn).}
}
\date{}
\maketitle

\begin{abstract}
In wireless networks, coded caching  is an effective technique to reduce network congestion during peak traffic times.
Recently, a new concept called placement delivery array (PDA) was proposed to characterize the  coded caching scheme.
So far, only one class of PDAs by Maddah-Ali and Niesen is known to be optimal. In this paper, we mainly focus on constructing optimal PDAs. Firstly,
we derive some lower bounds. Next, we present several infinite classes of PDAs, which are shown to be optimal with respect to the new bounds.
\end{abstract}
\section{Introduction}
Recently,  the explosive increasing mobile services, especially applications such as video streaming, have imposed a tremendous pressure on the data transmission over the core network. As a result, during the peak-traffic times, the communication systems are usually congested. Caching system, which proactively caches some contents at the network edge during off-peak hours, is a promising solution to alleviate  the  congestions % by shifting traffic from peak to off-peak times
 (see \cite{AA,GMDC,GGMG,JTLC,AN}, and references therein).

In order to further reduce the aforementioned  congestions, Maddah-Ali and Niesen proposed a coded caching approach based on network coding theory \cite{AN}. Particularly,
they focused on a $(K,F,M,N)$ caching system: a single server
containing $N$ files with the  same length $F$ packets connects to $K$ users over a shared link and each user has a cache memory of size $MF$ packets (see Fig. \ref{system}).
 \begin{figure}[htbp]
\centering\includegraphics[width=0.4\textwidth]{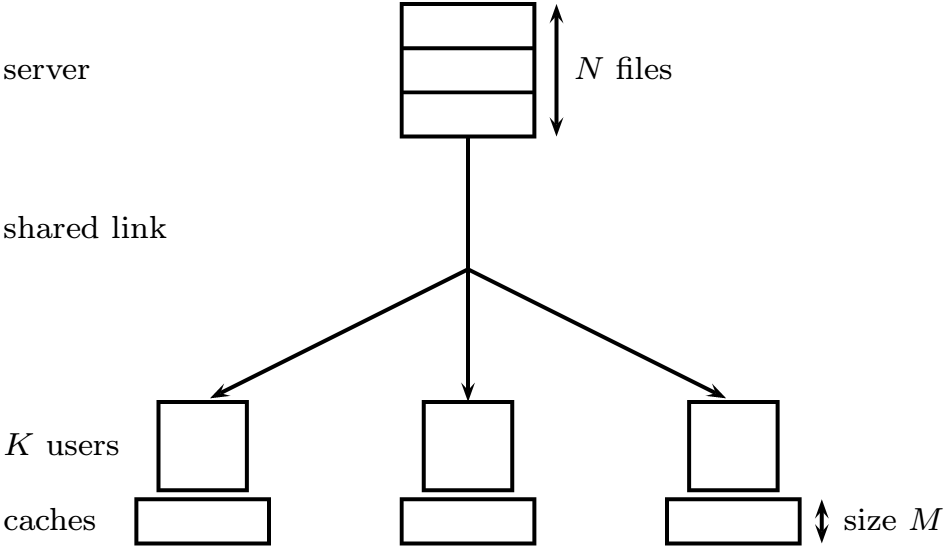}
\caption{Coded caching system}\label{system}
\end{figure}
A coded caching scheme consists of two phases:
\begin{itemize}
\item Placement phase during off-peak times: Parts of content are placed in users' cache memories independent of the user demands which are assumed to be arbitrary.
\item Delivery phase during peak times: Server sends a coded signal of at most $RF$ packets to the users such that each user's demand is satisfied.
\end{itemize}
 The objective is to minimize the  load $RF$ packets, in the delivery phase for the worst-case demands.
In the literature, $R$ is called the delivery rate.

In their seminal work \cite{AN}, Maddah-Ali and Niesen proposed a determined scheme for a $(K,F,M,N)$ coded
caching system, which is referred to as AN scheme in this paper.
Through an elaborate uncoded placement and a coded delivery  which could create multicast opportunity as much as possible,
the $(K,F,M,N)$ AN scheme can reduce the rate $R$ from $K(1-\frac{M}{N})$ of uncoded caching scheme to $K(1-\frac{M}{N})\frac{1}{1+\frac{KM}{N}}$.
Very recently,  it was shown in \cite{WTP} that  AN scheme has the minimum delivery rate under the constraint of uncoded cache placement.
So far, many results have been obtained following \cite{AN}, for instances, \cite{AG,CFL,GR,STC,T,WLG,WTP,YTC}  and so on.
In addition, AN scheme has been extensively employed in practical scenarios,
for examples, decentralized version \cite{MN1}, device to device networks \cite{JCM}, online
caching update \cite{RAN,YUTC} and hierarchical networks \cite{KNMD}, \cite{Xiao2016tree} etc.

In the placement phase of AN scheme, each file has to be  divided into $F={K\choose KM/N}$ packetes.
Clearly,  $F$ increases exponentially with the number of users $K$, which would make AN scheme infeasible when $K$ is large \cite{AN}. Therefore, designing a coded caching scheme with the minimum delivery rate for any integer $F$ becomes a critical issue, especially for practical implementations. In order to characterize the placement phase and delivery phase in a unified way, a new concept called $(K,F,Z,S)$ placement delivery array (PDA) was introduced in \cite{YCTC}, where $\frac{M}{N}=\frac{Z}{F}$ and $R=\frac{S}{F}$. Then, a $(K,F,M,N)$  caching system can be realized by a $(K,F,Z,S)$ PDA. Thus, it is very desirable to study optimal
$(K,F,Z,S)$ PDA with minimum $S$. In \cite{YCTC}, Yan \textit{et al.} showed that $(K,F,M,N)$ AN scheme is equivalent to a $(k,{k \choose t},{k-1 \choose t-1},$ ${k \choose t+1})$ PDA (AN PDA for short throughout this paper) with $k=K$ and $t=KM/N$, and presented two infinite classes of PDAs as well. Later on, two constructions of PDAs   in \cite{SZG,YTCC} were developed from \cite{YCTC}. Unfortunately,  only the AN PDA  was shown to be optimal among all the known constructions up to now.

In this paper, we concentrate on the constructions of optimal PDAs from a combinatorial viewpoint. We first derive lower bounds on the value of $S$.
With respect to the new lower bounds, the transpose of AN PDA is  shown to be optimal.  Then, we propose two kinds of recursive constructions, which lead to
optimal PDAs based on AN PDA and its transpose. Most notably, optimal PDAs with $Z=1$ and $F-1$ for any positive integers $K$ and $F$ are obtained, and some infinite classes of optimal PDAs for $Z=F-3,F-2$ are obtained.

The rest of this paper is organized as follows. Section \ref{preliminaries} introduces some preliminaries about PDAs. In Section \ref{lower-set}, lower bounds are derived. In Section \ref{sec-cons. optimal PDAs} and Section \ref{se-existence optimal PDA}, some constructions   of optimal PDAs are proposed. Finally, conclusion is drawn in Section \ref{c3conclusion}.

\section{Preliminaries}
\label{preliminaries}
Throughout this paper, the following notations are used.
\begin{itemize}
\item The arrays are denote by bold capital letters.
\item $[a,b]=\{a,a+1,\ldots,b\}$ and $[a, b)=\{a,a+1,\ldots,b-1\}$ for intervals of integers for any integers $a$ and $b$.
\item Given an array $\mathbf{P}=(a_{i,j})$ on $[0,S)\cup\{*\}$, set $\mathbf{P}+s=(a_{i,j}+s)$ where $s+*=*$ for any integer $s$.
%\item Given a square of order $F$, say $\mathbf{A}$, $\mathbf{A}|_{X}$ is a subarray obtained by deleting column $j\in [0,F)\setminus X$;
\end{itemize}

\begin{definition}\rm(\cite{YCTC})
\label{def-PDA}
For  positive integers $K,F$ and nonnegative integers $Z$ and $S$ with $F\ge Z$, an $F\times K$ array $\mathbf{P}=(p_{i,j})$, $i\in [0,F), j\in[0,K)$, composed of a specific symbol $``*"$ called star and $S$ nonnegative integers
$0,1,\cdots, S-1$, is called a $(K,F,Z,S)$ placement delivery array (PDA) if it satisfies the following conditions:
\begin{enumerate}
  \item [C$1$.] The symbol $``*"$ appears $Z$ times in each column;
  %\item [C2.] Each integer occurs at least once in the array;
  \item [C$2$.] For any two distinct entries $p_{i_1,j_1}$ and $p_{i_2,j_2}$,    $p_{i_1,j_1}=p_{i_2,j_2}=s$ is an integer  only if
  \begin{enumerate}
     \item [a.] $i_1\ne i_2$, $j_1\ne j_2$, i.e., they lie in distinct rows and distinct columns; and
     \item [b.] $p_{i_1,j_2}=p_{i_2,j_1}=*$, i.e., the corresponding $2\times 2$  subarray formed by rows $i_1,i_2$ and columns $j_1,j_2$ must be of the following form
  \begin{eqnarray*}
    \left(\begin{array}{cc}
      s & *\\
      * & s
    \end{array}\right)~\textrm{or}~
    \left(\begin{array}{cc}
      * & s\\
      s & *
    \end{array}\right).
  \end{eqnarray*}
   \end{enumerate}
  \end{enumerate}
\end{definition}

\begin{example}\rm
\label{E-pda}
It is easy to verify that, the following array is a $(4,6,3,4)$ PDA:
\begin{eqnarray*}
\mathbf{P}_{6\times 4}=\left(\begin{array}{cccc}
*&*&0&1\\
*&0&*&2\\
*&1&2&*\\
0&*&*&3\\
1&*&3&*\\
2&3&*&*
\end{array}\right).
\end{eqnarray*}
\end{example}

In \cite{YCTC}, Yan \textit{et al.}  showed that a $(K,F,Z,S)$ PDA $\mathbf{P}=(p_{i,j})_{F\times K}$ with $Z/F=M/N$ is
corresponding a $(K,F,M,N)$  caching scheme.
Precisely, each user is able to decode its requested file correctly for any request with delivery rate $R=S/F$. Generally speaking,
it would be preferred to construct a PDA with $S$ as small as possible for given positive integers $K,F,Z$, and thus we define
\begin{eqnarray}\label{Eqn_Def_S}
\mathcal{S}(K,F,Z) = \min_{(K,F,Z,S)~ \mathrm{PDA}} S
\end{eqnarray}
Then, a $(K,F,Z,S)$ PDA  is said
to be optimal if $S=\mathcal{S}(K,F,Z)$.

For each integer $t\in[0,K)$, let $F={K\choose t}$. Arrange all  the subsets with size $t+1$ of $[0,K)$ in the lexicographic order
and define $f_{t+1}(\Omega)$ to be its order minus 1 for any subset $\Omega$ of size $t+1$. Clearly,
$f_{t+1}$ is a bijection from $\{\Omega\subset[0,K):|\Omega|=t+1\}$ to $[0,{K\choose t+1})$.
Then, AN PDA is defined as a ${K\choose t}\times K$ array $\mathbf{P}=(p_{\mathcal{T},j})_{\mathcal{T}\subset [0,K),|\mathcal{T}|=t,j\in [0,K)}$
by
\begin{eqnarray}\label{Eqn_Def_AN}
p_{\mathcal{T},j}=\left\{\begin{array}{cc}
f_{t+1}(\mathcal{T}\cup\{j\}), & \mbox{if}~j\notin\mathcal{T}\\
*, & \mbox{otherwise}
\end{array}
\right.
\end{eqnarray}
where the rows are denoted  by the sets $\mathcal{T}\subset [0,K)$ and $|\mathcal{T}| = t$ \cite{YCTC}.
\begin{example}\label{ex-(4,6,3,4) PDA}
When $K=4$ and $t=2$, all the subsets of size $t+1=3$ in $\{0,1,2,3\}$ are ordered as
$\{0,1,2\},\{0,1,3\}$, $\{0$, $2$, $3\}$ and $\{1$, $2$, $3\}$, i.e.,
\begin{eqnarray*}
f_3(\{0,1,2\})=0,\ f_3(\{0,1,3\})=1,\ f_3(\{0,2,3\})=2 \hbox{ and } f_3(\{1,2,3\})=3.
\end{eqnarray*}
Then by \eqref{Eqn_Def_AN}, we have the $(4,6,3,4)$ PDA $\mathbf{P}_{6\times 4}$ in Example \ref{E-pda}.
\end{example}

It was shown in \cite{WTP} that the AN PDA has the minimum delivery rate under the constraint of uncoded cache placement. That is, $S/F$ is minimum,
which implies $\mathcal{S}(k,{k\choose t},{k-1\choose t-1})={k\choose t+1}$.

\begin{theorem}\label{th-AN-Y}
For any positive integers $k$ and $t$ with $0\le t\le k$, let ${k-1 \choose -1}={k\choose k+1}=0$. Then the  $(k,{k \choose t},{k-1 \choose t-1},{k\choose t+1})$ PDA defined in \eqref{Eqn_Def_AN} is optimal.
\end{theorem}

\section{Lower bounds on $S(K,F,Z)$}\label{lower-set}

Note that $\mathcal{S}(K,F,0)=KF$ is trivial. In this section, we derive lower bounds on $\mathcal{S}(K,F,Z)$ for positive integers $K$, $F$ and $Z$.

\begin{theorem}\label{th-digui-bound}
Given any positive integers $K,F,Z$ with $F\ge Z$,
\begin{eqnarray}\label{eq-digui-bound}
\mathcal{S}(K,F,Z)\geq \left\lceil\frac{(F-Z)K}{F}\right\rceil+\left\lceil\frac{F-Z-1}{F-1}\left\lceil\frac{(F-Z)K}{F}\right\rceil\right\rceil+\ldots+
\left\lceil\frac{1}{Z+1}\left\lceil\frac{2}{Z+2}\left\lceil\ldots\left\lceil\frac{(F-Z)K}{F}\right\rceil \ldots\right\rceil\right\rceil\right\rceil.
\end{eqnarray}
\end{theorem}

\begin{proof}
Since $\mathcal{S}(K,F,F)=0$ is clear, we only need to prove \eqref{eq-digui-bound} for $F>Z$.
Suppose that  there is a $(K,F,Z,\mathcal{S}(K,F,Z))$ PDA  $\mathbf{P}$.
Totally, there are $(F-Z)K$ integers in this array. Thus, among the $F$ rows,
there must exist one row containing at lest $\left\lceil \frac{(F-Z)K}{F} \right\rceil$ integers.
Without loss of generality, assume that these $\left\lceil \frac{(F-Z)K}{F} \right\rceil$ integers, say $0$, $1$, $\ldots$, $\left\lceil \frac{(F-Z)K}{F} \right\rceil-1$, are in the first row, i.e.,
\begin{eqnarray*}
\label{eq-digui}
\mathbf{P}=\left(\begin{array}{ccccc|c}
0      & 1     & 2     & \cdots & \left\lceil \frac{(F-Z)K}{F} \right\rceil-1 &\  \cdots\ \ \\ \hline
              &       &  &   &   &\\
       &       &&  \mathbf{P}' &   &\ \ddots\ \  \\
              &       &  &   &   &
                   \end{array}
  \right).
\end{eqnarray*}
If not, we can get such form  of $\mathbf{P}$ by row/column permutations.
Clearly $\mathbf{P}'$ is an $(\left\lceil \frac{(F-Z)K}{F} \right\rceil, F-1, Z, S')$ PDA for one nonnegative integer $S'$. Further, we
know that all the integer in  $[0, \left\lceil \frac{(F-Z)K}{F} \right\rceil)$
do not appear in $\mathbf{P}'$.  Otherwise, it would contradict Property C2. Therefore,  we have
\begin{eqnarray}\label{eq-diyiceng-bound}
\mathcal{S}(K,F,Z)\ge \left\lceil \frac{(F-Z)K}{F} \right\rceil+S'.
\end{eqnarray}

Perform  the same argument to $P'$ till that the remaining array is a $(\left\lceil\frac{1}{Z+1}\left\lceil\frac{2}{Z+2}\left\lceil\ldots\left\lceil\frac{(F-Z)K}{F}\right\rceil \ldots\right\rceil\right\rceil\right\rceil,Z,Z,0)$ PDA. Then, the bound in \eqref{eq-digui-bound} follws by recursively applying
\eqref{eq-diyiceng-bound} $F-Z$ times.
\end{proof}

When $Z=F-1$, \eqref{eq-digui-bound} can be written as
\begin{eqnarray}
\mathcal{S}(K,F,F-1)&\geq & \left\lceil\frac{K}{F}\right\rceil\label{eq-Z=F-1}.
\end{eqnarray}
We will show it is tight for any positive integers $K$ and $F$ in Subsection \ref{sub-Deleting in first recursive construction}. In the remainder of this section, we mainly concentrate on $Z<F-1$.

\begin{corollary}\label{cor-case-1}
For any positive integers $K,F$ and $Z$,
\begin{eqnarray}\label{eq-sepeical case}
\mathcal{S}(K,F,Z)\ge  \left\lceil\frac{(F-Z)K}{F}\right\rceil+F-Z-1.
\end{eqnarray}
\end{corollary}
\begin{proof}
The result directly follows from Theorem \ref{th-digui-bound} and the fact that
\begin{eqnarray}\label{Eqn_F_i}
\left\lceil\frac{F-Z-i}{F-i}\left\lceil\frac{F-Z-i+1}{F-i+1}\left\lceil\ldots\left\lceil\frac{(F-Z)K}{F}\right\rceil \ldots\right\rceil\right\rceil\right\rceil\ge 1
\end{eqnarray}
for any positive integer $1\leq i\leq F-Z-1$.
\end{proof}

Straightforwardly, it is seen from \eqref{Eqn_F_i} that the bound in \eqref{eq-sepeical case} is achievable only if
$\left\lceil\frac{(F-Z-1)}{F-1}\left\lceil\frac{(F-Z)K}{F}\right\rceil\right\rceil=1$.

\begin{example}\label{ex-(6,8,5,5)PDA}
When $K=6$, $F=8$ and $Z=5$, $\left\lceil\frac{(F-Z-1)}{F-1}\left\lceil\frac{(F-Z)K}{F}\right\rceil\right\rceil=1$. So we have $\mathcal{S}(6,8,5)\geq 3+3-1=5$. It is easy to check that the following array is an optimal $(6,8,5,5)$ PDA .
\begin{eqnarray*}
\mathbf{P}_{8\times 6}=\left(\begin{array}{cccccc}
 0    &*     &*    &*   &3   &*\\
 1    &3     &*    &*   &*   &4\\
 *    &0     &1    &*   &*   &*\\
 2    &*     &3    &*   &*   &*\\
 *    &2     &*    &1   &*   &*\\
 *    &*     &4    &0   &2   &*\\
 *    &*     &*    &*   &1   &0\\
 *    &*     &*    &3   &*   &2
\end{array}\right)
\end{eqnarray*}
\end{example}

We can further improve the bound in \eqref{eq-sepeical case} in some cases with the help of the following lower bound.
\begin{lemma}\label{th-cha3-6}
For any positive integers $K,F$ and $Z$ with $F-1>Z$,
\begin{eqnarray}
\label{eq-mod3}
\mathcal{S}\Big(\left\lceil\frac{(F-Z)K}{\mathcal{S}(K,F,Z)}\right\rceil, F-\left\lceil\frac{(F-Z)K}{\mathcal{S}(K,F,Z)}\right\rceil,Z-\left\lceil\frac{(F-Z)K}{\mathcal{S}(K,F,Z)}\right\rceil+1\Big)\le \mathcal{S}(K,F,Z)-1.
\end{eqnarray}
\end{lemma}

\begin{proof} Assume that there exists a $(K,F,Z,\mathcal{S}(K,F,Z))$ PDA  $\mathbf{P}$. In this array,
there are $\mathcal{S}(K,F,Z)$ distinct ones among all the  $(F-Z)K$ integers. Then, there must exist an
integer $s\in [0,\mathcal{S}(K,F,Z))$ occurring at least $\left\lceil \frac{(F-Z)K}{\mathcal{S}(K,F,Z)}\right\rceil$ times in $\mathbf{P}$. By means of row/column permutations,
we are able to write $\mathbf{P}$ as follows
\begin{eqnarray}
\label{eq-pda-mod3}
\mathbf{P}=\left(\begin{array}{ccccc|c}
s      & *     & *     & \ldots & * &\  \ldots\ \ \\
*      & s     & *     & \ldots & * &\ \ldots\ \ \\
\vdots &       &\vdots &        &   &\ \ddots\ \ \\
*      & *     & *     & \ldots & s &\ \ldots\ \ \\
\hline\\[-0.5cm]
       &       &  &   &   &\\
       &       & \mathbf{P}'&   &   &\ \ddots\ \  \\
              &       &  &   &   &
                   \end{array}
  \right)
\end{eqnarray}
Obviously, $\mathbf{P}'$ is a $(\left\lceil\frac{(F-Z)K}{\mathcal{S}(K,F,Z)}\right\rceil, F-\left\lceil\frac{(F-Z)K}{\mathcal{S}(K,F,Z)}\right\rceil,Z-\left\lceil\frac{(F-Z)K}{\mathcal{S}(K,F,Z)}\right\rceil+1,S')$ PDA with $S'\le \mathcal{S}(K,F,Z)-1$,
which implies
\begin{eqnarray*}
S\Big(\left\lceil\frac{(F-Z)K}{\mathcal{S}(K,F,Z)}\right\rceil, F-\left\lceil\frac{(F-Z)K}{\mathcal{S}(K,F,Z)}\right\rceil,Z-\left\lceil\frac{(F-Z)K}{\mathcal{S}(K,F,Z)}\right\rceil+1\Big)\leq S'\le
\mathcal{S}(K,F,Z)-1.
\end{eqnarray*}
\end{proof}

\begin{theorem}\label{th-furhter improved}
If positive integers $K,F,Z$ satisfying $\left\lceil\frac{(F-Z-1)}{F-1}\left\lceil\frac{(F-Z)K}{F}\right\rceil\right\rceil=1$ and
$\left\lceil\frac{(F-Z)K}{F}\right\rceil F<\left\lceil\frac{(F-Z)K}{\left\lceil\frac{(F-Z)K}{F}\right\rceil+F-Z-1}\right\rceil $ $(\left\lceil\frac{(F-Z)K}{F}\right\rceil+F-Z-1)$, then
\begin{eqnarray}\label{Eqn_bound_1}
\mathcal{S}(K,F,Z)\geq \left\lceil\frac{(F-Z)K}{F}\right\rceil+F-Z
\end{eqnarray}
\end{theorem}

\begin{proof} We prove this statement by contradiction. If \eqref{Eqn_bound_1}
does not hold, then by \eqref{eq-sepeical case}
\begin{eqnarray}\label{Eqn_S1}
\mathcal{S}(K,F,Z)=\left\lceil\frac{(F-Z)K}{F}\right\rceil+F-Z-1
\end{eqnarray}

By \eqref{eq-mod3} and \eqref{Eqn_S1},
\begin{eqnarray*}
&&\left(\left\lceil\frac{(F-Z)K}{F}\right\rceil+(F-Z)-1\right)-1\\
&=&\mathcal{S}(K,F,Z)-1\\
&\geq& \mathcal{S}\Big(\left\lceil\frac{(F-Z)K}{\mathcal{S}(K,F,Z)}\right\rceil, F-\left\lceil\frac{(F-Z)K}{\mathcal{S}(K,F,Z)}\right\rceil,Z-\left\lceil\frac{(F-Z)K}{\mathcal{S}(K,F,Z)}\right\rceil+1\Big)\\
&\geq&\left\lceil\frac{(F-\left\lceil\frac{(F-Z)K}{\mathcal{S}(K,F,Z)}\right\rceil-(Z-\left\lceil\frac{(F-Z)K}{\mathcal{S}(K,F,Z)}\right\rceil+1)) \left\lceil\frac{(F-Z)K}{\mathcal{S}(K,F,Z)}\right\rceil}{F-\left\lceil\frac{(F-Z)K}{\mathcal{S}(K,F,Z)}\right\rceil}\right\rceil+\left(F-
\left\lceil\frac{(F-Z)K}{\mathcal{S}(K,F,Z)}\right\rceil-\left(Z-\left\lceil\frac{(F-Z)K}{\mathcal{S}(K,F,Z)}\right\rceil+1\right)\right)-1\\
&=&\left\lceil\frac{(F-Z-1) \left\lceil\frac{(F-Z)K}{\mathcal{S}(K,F,Z)}\right\rceil}{F-\left\lceil\frac{(F-Z)K}{\mathcal{S}(K,F,Z)}\right\rceil}\right\rceil+F-Z-2
\end{eqnarray*}
Note that in the second inequality, we apply the lower bound \eqref{eq-sepeical case}
to the $(\left\lceil\frac{(F-Z)K}{\mathcal{S}(K,F,Z)}\right\rceil, F-\left\lceil\frac{(F-Z)K}{\mathcal{S}(K,F,Z)}\right\rceil,Z-\left\lceil\frac{(F-Z)K}{\mathcal{S}(K,F,Z)}\right\rceil+1,S')$ PDA
$\mathbf{P}'$ in \eqref{eq-pda-mod3}.

Therefore, we have
\begin{eqnarray*}
 \left\lceil\frac{(F-Z)K}{F}\right\rceil\geq \frac{(F-Z-1) \left\lceil\frac{(F-Z)K}{\mathcal{S}(K,F,Z)}\right\rceil}{F-\left\lceil\frac{(F-Z)K}{\mathcal{S}(K,F,Z)}\right\rceil}.
\end{eqnarray*}
Let $\left\lceil\frac{(F-Z)K}{F}\right\rceil F=(F-Z)K+ \gamma$ and $\left\lceil\frac{(F-Z)K}{\mathcal{S}(K,F,Z)}\right\rceil \mathcal{S}(K,F,Z)=(F-Z)K+\delta$.
Then,
\begin{eqnarray*}
\frac{(F-Z)K+\gamma}{F}\geq\frac{(F-Z-1) \frac{(F-Z)K+\delta}{\mathcal{S}(K,F,Z)}}{F-\frac{(F-Z)K+\delta}{\mathcal{S}(K,F,Z)}}.
\end{eqnarray*}
which gives
\begin{eqnarray*}
\frac{F-\frac{(F-Z)K+\delta}{\mathcal{S}(K,F,Z)}}{F}\geq \frac{(F-Z-1) \frac{(F-Z)K+\delta}{\mathcal{S}(K,F,Z)}}{(F-Z)K+\gamma},
\end{eqnarray*}
i.e.,
\begin{eqnarray*}
\mathcal{S}(K,F,Z)-\frac{(F-Z)K+\delta}{F}\geq (F-Z-1)+\frac{(F-Z-1)(\delta-\gamma)}{(F-Z)K+\gamma}.
\end{eqnarray*}

Substituting \eqref{Eqn_S1} into above equation, we can obtain
\begin{eqnarray*}
(\delta-\gamma)\left(\frac{F-Z-1}{(F-Z)K+\gamma}+\frac{1}{F}\right)\leq 0.
\end{eqnarray*}
which implies that $\delta-\gamma\leq 0$, i.e., $\left\lceil\frac{(F-Z)K}{\mathcal{S}(K,F,Z)}\right\rceil \mathcal{S}(K,F,Z)-(F-Z)K \leq \left\lceil\frac{(F-Z)K}{F}\right\rceil F -(F-Z)K $.
Again applying \eqref{Eqn_S1}, we then have $\left\lceil\frac{(F-Z)K}{\left\lceil\frac{(F-Z)K}{F}\right\rceil+F-Z-1}\right\rceil(\left\lceil\frac{(F-Z)K}{F}\right\rceil+F-Z-1)
\le \left\lceil\frac{(F-Z)K}{F}\right\rceil F$.
This  contradicts  the hypothesis.
\end{proof}

The following improved bounds can be obtained immediately from Theorem \ref{th-furhter improved}.

\begin{corollary}\label{co-tighter S(K,F,F-2)}
If  $K\leq\frac{F(F-1)}{2}$ and $\left\lceil\frac{2K}{F} \right\rceil F<\left\lceil \frac{2K}{\left\lceil\frac{2K}{F} \right\rceil+1}\right\rceil(\left\lceil\frac{2K}{F} \right\rceil+1)$, then
\begin{eqnarray}
\mathcal{S}(K,F,F-2)\geq \left\lceil\frac{2K}{F} \right\rceil+2
\end{eqnarray}
\end{corollary}

\begin{corollary}\label{co-K=F lower bound}
If $(2F-2Z-1)\nmid(F-Z)F$, then $\mathcal{S}(F,F,Z)\geq 2(F-Z)$.
\end{corollary}
\begin{proof}
If $\left\lceil\frac{(F-Z-1)(F-Z)}{F-1}\right\rceil=1$ and $(2F-2Z-1)\nmid(F-Z)F$, the result directly follows from Theorem \ref{th-furhter improved}.
If $\left\lceil\frac{(F-Z-1)(F-Z)}{F-1}\right\rceil\geq 2$, similarly to the proof of Corollary \ref{cor-case-1}, it follows from Theorem \ref{th-digui-bound}.
\begin{eqnarray*}
\mathcal{S}(F,F,Z)\geq (F-Z)+\left\lceil\frac{(F-Z-1)(F-Z)}{F-1}\right\rceil+(F-Z-2)=2(F-Z).
\end{eqnarray*}

\end{proof}

\section{Constructions of optimal PDAs}
\label{sec-cons. optimal PDAs}
Given any positive integer $k$ and nonnegative integer $t$ with $t\leq k$, in this section we begin from the optimal $(k,{k \choose t},{k-1 \choose t-1},{k \choose t+1})$ PDA
$\mathbf{P}$ in \eqref{Eqn_Def_AN} to present some optimal PDAs.
First of all, consider its transpose $\mathbf{P}^\top=(p'_{j,\mathcal{T}})_{j\in [0,k),\mathcal{T}\subset [0,k),|\mathcal{T}|=t}$, i.e.,
\begin{eqnarray}\label{Eqn_Def_AN_T}
p'_{j,\mathcal{T}}=p_{\mathcal{T},j}=\left\{\begin{array}{cc}
f_{t+1}(\mathcal{T}\cup\{j\}), & \mbox{if}~j\notin\mathcal{T}\\
*, & \mbox{otherwise}
\end{array}
\right.
\end{eqnarray}
where the column is denoted by the sets in $\mathcal{T}\subset [0,k)$ and $|\mathcal{T}| = t$.
\begin{example}\label{ex-(6,4,2,4) PDA}
When $k=4$ and $t=2$, a $(6,4,2,4)$ PDA is obtained from \eqref{Eqn_Def_AN_T} as
\begin{eqnarray*}
\mathbf{P}^\top_{6\times 4}=\left(\begin{array}{cccccc}
*&*&*& 0& 1& 2\\
*& 0& 1&*&*& 3\\
0&*& 2&*& 3&*\\
1& 2&*& 3&*&*
\end{array}\right),
\end{eqnarray*}
which is the transpose of the PDA in Example \ref{ex-(4,6,3,4) PDA}.
\end{example}

It is easy to check that $\mathbf{P}^\top$ is a $({k \choose t},k,t,{k\choose t+1})$ PDA. We show its optimality, i.e.,
$\mathcal{S}({k \choose t},k,t)={k\choose t+1}$, in the following theorem.

\begin{theorem}\label{th-Ali-optimal}
For any positive integer $k$ and nonnegative integer $t$ with $t\le k$, $\mathbf{P}^\top$ in \eqref{Eqn_Def_AN_T} is an optimal $({k \choose t},k,t,{k\choose t+1})$ PDA.
\end{theorem}

\begin{proof}
According to \eqref{eq-digui-bound},
\begin{eqnarray}\label{Eqn_sm}
&&\mathcal{S}\left({k \choose t},k,t\right)\nonumber\\
&\geq&  \left\lceil\frac{(k-t){k \choose t}}{k}\right\rceil+\left\lceil\frac{(k-t-1)}{k-1}\left\lceil\frac{(k-t){k \choose t}}{k}\right\rceil\right\rceil+\ldots+
\left\lceil\frac{1}{t+1}\left\lceil\frac{2}{t+2}\left\lceil\ldots\left\lceil\frac{(k-t){k \choose t}}{k}\right\rceil \ldots\right\rceil\right\rceil\right\rceil\nonumber\\
&=&{k-1\choose k-t-1}+{k-2\choose k-t-2}+\ldots+{k-(k-t-1)\choose k-t-(k-t-1)}+{k-(k-t)\choose k-t-(k-t)}\nonumber\\
&=&{k-1\choose t}+{k-2\choose t}+\ldots+{t+1\choose t}+{t\choose t}\\
&=&{k\choose t+1}\nonumber
\end{eqnarray}
where last identity holds due to ${n\choose t+1}={n-1 \choose t}+{n-1\choose t+1}$
for any positive integers $n\ge t+2$ and the fact ${t+1\choose t+1}={t\choose t}$.
It then follows from \eqref{Eqn_Def_S} that $\mathcal{S}({k \choose t},k,t)={k\choose t+1}$ for $k\ge t+2$.

Additionally, it is easy to see that $\mathcal{S}(1,k,k)={k\choose k+1}=0$ and $\mathcal{S}(k,k,k-1)={k\choose k}=1$. That is,
the assertion is also true for  $t=k$ or $k-1$.

\end{proof}

In what follows, we propose two kinds of recursive constructions based on the optimal PDAs in Theorems \ref{th-AN-Y} and
\ref{th-Ali-optimal}. By means of these recursive constructions, several infinite classes of optimal PDAs can be obtained.

\subsection{The first recursive construction}\label{sub-first recursive}

\begin{construction}\label{con-first recursive}
Given positive  integers $K_i,F$ and $Z$, let $\mathbf{P}_i$ be a $(K_i,F,Z,S_i)$ PDA with integers $[0,S_i)$ for $0\le i <m$ where $m\geq 2$. Generate an $F\times mK$ array
\begin{eqnarray}\label{eq-first-recursive}
\mathbf{P}=\Big(\mathbf{P}_0,\mathbf{P}_1+S_0,\cdots, \mathbf{P}_{m-1}+\sum_{i=0}^{m-2}S_i\Big).
\end{eqnarray}
It is not difficult to check from the definition that $\mathbf{P}$ in \eqref{eq-first-recursive} is an $(\sum_{i=0}^{m-1}K,F,Z,\sum_{i=0}^{m-1}S_i)$ PDA.
\end{construction}

\begin{theorem}\label{th-optimal-frist Ali}
For any positive integers $k$, $t$ and $m$ with $0\leq t\leq k$, there exists an optimal $(m{k\choose t},k,t,m{k\choose t+1})$ PDA.
\end{theorem}

\begin{proof}
Setting $\mathbf{P}_i$ be the same optimal $({k\choose t},k,t,{k\choose t+1})$ PDA given in Theorem \ref{th-Ali-optimal}
for all $0\le i<m$,  we get an $(m{k\choose t},k,t,m{k\choose t+1})$ PDA
by Construction \ref{con-first recursive}. Hereafter, it suffices to show $\mathcal{S}(m{k\choose t},k,t)\geq m{k\choose t+1}$.

Applying \eqref{eq-digui-bound} and \eqref{Eqn_sm}, we have
\begin{eqnarray*}
&&\mathcal{S}\left(m{k \choose t},k,t\right)\\
&\geq&  \left\lceil\frac{(k-t)m{k \choose t}}{k}\right\rceil+\left\lceil\frac{(k-t-1)}{k-1}\left\lceil\frac{(k-t)m{k \choose t}}{k}\right\rceil\right\rceil+\cdots+
\left\lceil\frac{1}{t+1}\left\lceil\frac{2}{t+2}\left\lceil\ldots\left\lceil\frac{(k-t)m{k \choose t}}{k}\right\rceil \ldots\right\rceil\right\rceil\right\rceil\\
&\geq&m\left\{ \left\lceil\frac{(k-t){k \choose t}}{k}\right\rceil+\left\lceil\frac{(k-t-1)}{k-1}\left\lceil\frac{(k-t){k \choose t}}{k}\right\rceil\right\rceil
+\cdots+
\left\lceil\frac{1}{t+1}\left\lceil\frac{2}{t+2}\left\lceil\ldots\left\lceil\frac{(k-t){k \choose t}}{k}\right\rceil \ldots\right\rceil\right\rceil\right\rceil\right\}\\
&=&m{k\choose t+1}.
\end{eqnarray*}
\end{proof}
\begin{example}
\label{ex-(12,4,2,8) PDA}
Let $m=2$, $k=4$ and $t=2$. From Example \ref{ex-(6,4,2,4) PDA} and Theorem \ref{th-optimal-frist Ali}, an optimal $(12,4,2,8)$ PDA can be obtained as follows.
\begin{eqnarray*}
\mathbf{P}_{4\times 12}=\left(\begin{array}{cccccc|cccccc}
*&* &* &0 &1 &2  &* &* &* &4 &5 &6\\
*&0 &1 &* &* &3  &* &4 &5 &* &* &7\\
0&* &2 &* &3 &*  &4 &* &6 &* &7 &*\\
1&2 &* &3 &* &*  &5 &6 &* &7 &* &*
\end{array}\right).
\end{eqnarray*}
\end{example}

\subsection{The second recursive construction}\label{sub-second recursive}

\begin{construction}\label{con-second recursive}
Let $\mathbf{P}_i$ be  $(K_i,F_i,Z_i,S_i)$ PDA, $i=0,1,\ldots, m-1$ and $m\geq 2$. Generate an $\sum_{i=0}^{m-1}F_i\times \sum_{i=0}^{m-1}K_i$ array
\begin{eqnarray}
\label{eq-recursive construction}
\mathbf{P}=\left(\begin{array}{cccc}
\mathbf{P}_{0}&*              &\ldots   & *               \\
*             &\mathbf{P}_{1} &\ldots   & *               \\
\vdots        &\vdots         &\ddots   &\vdots               \\
*             &*              &\cdots   &\mathbf{P}_{m-1}
\end{array}\right).
\end{eqnarray}
It is not difficult to verify from  the definition  that $\mathbf{P}$ is a $(\sum_{i=0}^{m-1}K_i,\sum_{i=0}^{m-1}F_i,\sum_{i=0}^{m-2}F_i+Z_{m-1},S)$ PDA
with $S=\max \{S_i, 0\le i<m\}$ if $F_i-Z_i=F_0-Z_0$ for any $0\le i<m$.
\end{construction}

\begin{lemma}\label{co-recursive-optimal}
Let $\mathbf{P}_i$ be an optimal $(K_i,F_i,Z_i,S_i)$ PDA with $S_i= \left\lceil\frac{(F_i-Z_i)K_i}{F_i}\right\rceil+(F_i-Z_i)-1$ for all $0\le i<m$ and $m\ge 2$. If $F_i-Z_i=F_0-Z_0$ and $S=S_i$ for any $0\le i<m$, then $\mathbf{P}$ in \eqref{eq-recursive construction} is an  optimal $(\sum_{i=0}^{m-1}K_i,\sum_{i=0}^{m-1}F_i,\sum_{i=0}^{m-2}F_i+Z_{m-1},S)$ PDA.
\end{lemma}

\begin{proof} In order to prove the optimality of $\mathbf{P}$, it is sufficient to show  $\mathcal{S}(\sum_{i=0}^{m-1}K_i$, $\sum_{i=0}^{m-1}F_i$, $\sum_{i=0}^{m-2}F_i+Z_{m-1})\geq S$ by \eqref{Eqn_Def_S}. According to \eqref{eq-sepeical case},
\begin{eqnarray}\label{eq-fuhe-bound1}
&&\mathcal{S}\left(\sum_{i=0}^{m-1}K_i,\sum_{i=0}^{m-1}F_i,\sum_{i=0}^{m-2}F_i+Z_{m-1}\right) \nonumber\\
&\geq &\left\lceil\frac{(\sum_{i=0}^{m-1}F_i-(\sum_{i=0}^{m-2}F_i+Z_{m-1}))\sum_{i=0}^{m-1}K_i}{\sum_{i=0}^{m-1}F_i}
\right\rceil+\left(\sum_{i=0}^{m-1}F_i-(\sum_{i=0}^{m-2}F_i+Z_{m-1})\right)-1\nonumber\\
&=&\left\lceil\frac{(F_{m-1}-Z_{m-1})\sum_{i=0}^{m-1}K_i}{\sum_{i=0}^{m-1}F_i}\right\rceil+ (F_{m-1}-Z_{m-1})-1
\end{eqnarray}
There exists some integers $0\le j<m$ such that $\frac{K_j}{F_j}=\min_{0\le i<m} \frac{K_i}{F_i}$.
Then, we have
\begin{eqnarray}\label{eq-fuhe-bound2}
\frac{K_0+K_1+\ldots+K_{m-2}+K_{m-1}}{F_0+F_1+\ldots+F_{m-2}+F_{m-1}}\geq\frac{K_{j}}{F_{j}}
\end{eqnarray}

Substituting \eqref{eq-fuhe-bound2} into \eqref{eq-fuhe-bound1}, we obtain
\begin{eqnarray*}
\mathcal{S}\left(\sum_{i=0}^{m-1}K_i,\sum_{i=0}^{m-1}F_i,\sum_{i=0}^{m-2}F_i+Z_{m-1}\right)\geq  \left\lceil\frac{(F_j-Z_j)K_j}{F_j}\right\rceil+(F_j-Z_j)-1=S
\end{eqnarray*}
where in the last identity we use the fact that $S=S_j= \left\lceil\frac{(F_j-Z_j)K_j}{F_j}\right\rceil+(F_j-Z_j)-1$ and $F_j-Z_j=F_{m-1}-Z_{m-1}$.
\end{proof}

Consider the optimal $(k,{k \choose t},{k-1\choose t-1},{k\choose t+1})$ PDA in Theorem \ref{th-AN-Y} and
optimal $({k \choose t},k,t,{k\choose t+1})$ PDA in Theorem \ref{th-Ali-optimal}.  It can be seen that
they satisfy $\mathcal{S}(K,F,Z)=\left\lceil\frac{(F-Z)K}{F}\right\rceil+(F-Z)-1$ if $t=0$ or $t=k-2$,
which are corresponding to $(k,1,0,k)$ PDA, $(1,k,0,k)$ PDA, $(k, {k\choose 2},{k-1\choose 2},k)$  PDA, and $({k\choose 2},k,k-2,k)$  PDA respectively.
Applying Construction \ref{con-second recursive} to these PDAs, we get the following optimal ones by Lemma  \ref{co-recursive-optimal}.
\begin{theorem}\label{co-Z=0}
For any positive integers $k$ and $m$, there exists 1) optimal $(mk,m,m-1,k)$ PDA; 2) optimal $(m,mk,(m-1)k,k)$ PDA;
3) optimal $(mk$, $m{k\choose 2}$, $m{k\choose 2}-k+1$, $k)$ PDA ; and 4) optimal $(m{k\choose 2},mk,mk-2,k)$ PDA.
\end{theorem}

\begin{example}
\label{ex-three optimal PDAs}
1) When $m=2$ and $k=2$, from Theorem \ref{co-Z=0}-1) and Theorem \ref{co-Z=0}-2), an optimal $(4,2,1,2)$ PDA and an optimal $(2,4,2,2)$ PDA can be obtained as follows by using $\mathbf{P}_{1\times 2}$ generated by \eqref{Eqn_Def_AN} and $\mathbf{P}^\top_{1\times 2}$ generated by \eqref{Eqn_Def_AN_T} respectively.
\begin{eqnarray*}
\mathbf{P}_{1\times 2}=\left(\begin{array}{cc}
0& 1
\end{array}\right)\ \
\mathbf{P}_{2\times 4}=\left(\begin{array}{cc|cc}
0&1&*&*\\
\hline
*&*&0&1
\end{array}\right)\ \ \ \ \ \ \ \ \ \ \ \
\mathbf{P}^\top_{1\times 2}=\left(\begin{array}{c}
0\\
1
\end{array}\right)\ \
\mathbf{P}_{4\times 2}=\left(\begin{array}{c|c}
0&*\\
1&*\\
\hline
*&0\\
*&1
\end{array}\right)
\end{eqnarray*}

2) When $m=2$ and $k=4$, from Theorem \ref{co-Z=0}-3) and Theorem \ref{co-Z=0}-4), an optimal $(8,12,9,4)$ PDA and an optimal $(12,8,6,4)$ PDA can be obtained in the following by using $\mathbf{P}_{6\times 4}$ in Example \ref{E-pda} and $\mathbf{P}^\top_{6\times 4}$ in Example \ref{ex-(6,4,2,4) PDA} respectively.
\begin{eqnarray*}
\label{Epdm}
\mathbf{P}_{12\times 8}=\left(\begin{array}{cccc|cccc}
*&*&0&1&*&*&*&*\\
*&0&*&2&*&*&*&*\\
*&1&2&*&*&*&*&*\\
0&*&*&3&*&*&*&*\\
1&*&3&*&*&*&*&*\\
2&3&*&*&*&*&*&*\\
\hline
*&*&*&*&*&*&0&1\\
*&*&*&*&*&0&*&2\\
*&*&*&*&*&1&2&*\\
*&*&*&*&0&*&*&3\\
*&*&*&*&1&*&3&*\\
*&*&*&*&2&3&*&*
\end{array}\right)\ \ \ \ \
\mathbf{P}_{8\times 12}=\left(\begin{array}{cccccc|cccccc}
*&*&*&0&1&2&*&*&*&*&*&*\\
*&0&1&*&*&3&*&*&*&*&*&*\\
0&*&2&*&3&*&*&*&*&*&*&*\\
1&2&*&3&*&*&*&*&*&*&*&*\\
\hline
*&*&*&*&*&*&*&*&*&0&1&2\\
*&*&*&*&*&*&*&0&1&*&*&3\\
*&*&*&*&*&*&0&*&2&*&3&*\\
*&*&*&*&*&*&1&2&*&3&*&*
\end{array}\right)
\end{eqnarray*}
\end{example}

Moreover, one special case of $S_i\geq \left\lceil\frac{(F_i-Z_i)K_i}{F_i}\right\rceil+(F_i-Z_i)-1$ in
Construction \ref{con-second recursive} can be discussed in the following theorem.

\begin{theorem}\label{co-existance Z=F-3}
For any positive integer $F>10$ with $5\nmid F$, there exists an optimal $(F,F,F-3,6)$ PDA .
\end{theorem}

\begin{proof}  Since  $5 \nmid F$, it follows from Corollary \ref{co-K=F lower bound} that
\begin{eqnarray}
\label{eq-Z=F-3-bound}
\mathcal{S}(F,F,F-3)\geq 6
\end{eqnarray}
So the following PDAs are optimal.
\begin{eqnarray*}
\label{eq-K=F}
\mathbf{P}_{4\times 4}=\left(\begin{array}{cccc}
0 & 3 & 5 & *\\
1 & 4 & * & 5\\
2 & * & 4 & 3\\
* & 2 & 1 & 0\\
\end{array}\right)\ \ \
\mathbf{P}_{6\times 6}=\left(\begin{array}{cccccc}
0 & 3 & 5 & * & * & *\\
1 & 4 & * & 5 & * & *\\
2 & * & 4 & * & 3 & *\\
* & 2 & * & * & 0 & 5\\
* & * & 1 & 0 & * & 3\\
* & * & * & 2 & 1 & 4
\end{array}\right)\ \ \
\mathbf{P}_{7\times 7}=\left(\begin{array}{ccccccc}
0 & 3 & 5 & * & * & * & *\\
1 & 4 & * & 5 & * & * & *\\
2 & * & * & * & 3 & 5 & *\\
* & 2 & * & * & 0 & * & 5\\
* & * & 1 & 0 & * & * & 4\\
* & * & 2 & * & * & 0 & 3\\
* & * & * & 2 & 1 & 4 & *
\end{array}\right)
\end{eqnarray*}
Herein, we use above optimal PDAs $\mathbf{P}_{4\times 4}$, $\mathbf{P}_{6\times 6}$ and $\mathbf{P}_{7\times 7}$ to
construct infinite optimal $(F,F,F-3,6)$ PDAs.
\begin{itemize}
  \item[(i)] Construct $(4m,4m,4m-3,6)$ PDA by applying Construction \ref{con-second recursive}
  to $\mathbf{P}_i = \mathbf{P}_{4\times 4}$, where $0 \leq i \leq m-1$.
  \item[(ii)] Generate $(4m+6,4m+6,4m+3,6)$ PDA ($(4m+7,4m+7,4m+4,6)$ PDA respectively)
  by setting $\mathbf{P}_i = \mathbf{P}_{4\times 4}$, $0 \leq i \leq m-1$,
  and $\mathbf{P}_m = \mathbf{P}_{6\times 6}$ ($\mathbf{P}_m = \mathbf{P}_{7\times 7}$ respectively) in
  Construction \ref{con-second recursive}.
  \item [(iii)]  Construct $(4m+13,4m+13,4m+10,6)$ PDA by applying Construction \ref{con-second recursive}
  to $\mathbf{P}_i = \mathbf{P}_{4\times 4}$, $0 \leq i \leq m-1$, $\mathbf{P}_m = \mathbf{P}_{6\times 6}$,
  $\mathbf{P}_{m+1} = \mathbf{P}_{7\times 7}$.
\end{itemize}
\end{proof}

In addition, it is possible to get more optimal PDAs by combining Constructions \ref{con-first recursive}
and \ref{con-second recursive}.

\begin{corollary}\label{co-comb}
There exists an optimal $(\frac{n^2}{2}, n,n-2,n+2)$ for any positive even integer $n$.
\end{corollary}

\begin{proof}
Set $\mathbf{P}_0$ to be the optimal $(\frac{n(n-1)}{2}, n,n-2,n)$ PDA  with integer set $\{0,1,\cdots,n-1\}$
obtained by Theorem \ref{th-Ali-optimal} in place of $k=n$ and $t=n-2$,
and $\mathbf{P}_1$ to be the optimal $(\frac{n}{2}, n,n-2,2)$ PDA  with integer set $\{0,1\}$
obtained by Theorem \ref{co-Z=0}-2) in place of $m=n/2$ and $k=2$.
Then, we can construct an  $(\frac{n^2}{2}, n,n-2,n+2)$ PDA  $\mathbf{P}=(\mathbf{P}_0,\mathbf{P}_1+n)$
based on Construction \ref{con-first recursive}, which is optimal by \eqref{eq-digui-bound}.
\end{proof}

\begin{example}\rm
\label{ex-h-F=Z-2}
According to Corollary \ref{co-comb}, when $n=4$, we can construct an optimal $(8, 4,2,6)$ PDA by
using $\mathbf{P}^\top_{6\times 4}$ in Example \ref{ex-(6,4,2,4) PDA} and  $\mathbf{P}_{4\times 2}$ in  Example \ref{ex-three optimal PDAs}-1).
\begin{eqnarray*}
\mathbf{P}_{4\times 8}=\left(\begin{array}{cccccc|ccc}
 *&* &* &0 &1 &2  &4 &* \\
*&0 &1 &* &* &3  &5 &* \\
0&* &2 &* &3 &*  &* &4\\
1&2 &* &3 &* &*  &* &5
 \end{array}\right)
\end{eqnarray*}
\end{example}

\section{Further results on optimal PDAs}
\label{se-existence optimal PDA}

First of all, we state a simple fact about a $(K,F,Z,S)$ PDA.
\begin{fact}
The resultant array is still a $(K-\kappa,F,Z,S')$ PDA for an integer $0<S'\le S$ if we delete $\kappa$ columns.
\end{fact}

In the sequel, we present more optimal PDAs by applying this fact  to the two aforementioned recursive constructions respectively.

\subsection{Deleting method in the first recursive construction}
\label{sub-Deleting in first recursive construction}
\begin{theorem}\label{co-Z=F-1}
There exists an optimal $(K,F,F-1,\lceil {K\over F}\rceil)$ PDA  for any positive integers $K$ and $F$.
\end{theorem}

\begin{proof} Assume that $K=(m-1)F+\kappa$ for integers $m\ge 1$ and $0\le \kappa<F$. Let $\mathbf{P}$ be the optimal  $(mF,F,F-1,m)$ PDA
obtained by Theorem \ref{th-optimal-frist Ali} in place of $k=F$ and $t=F-1$.
Note that the original $(F,F,F-1,1)$ contains only one integer and thus
the last $F$ columns of $\mathbf{P}$ in \eqref{eq-first-recursive} only have one integer. Then, we can get a $(K,F,F-1,m)$ PDA if
delete the last $F-\kappa$ columns, which is optimal by Corollary \ref{cor-case-1} since $\mathcal{S}(K,F,Z)=\lceil\frac{mF-\kappa}{F}\rceil=m$.
\end{proof}

\begin{example}\rm
\label{ex-Z=F-1}
From Theorem \ref{co-Z=F-1}, when $F=3$ and $K=5$,
we can first obtain an optimal $(6,3,2,2)$ PDA  by Theorem \ref{th-optimal-frist Ali} in place of $m=2$, $k=3$ and $t=2$ as follows.
\begin{eqnarray*}
\mathbf{P}_{3\times 6}=\left(\begin{array}{ccc|ccc}
 0     &* &* &1      &* &* \\
 *&0      &* &* &1     &* \\
 *&* &0      &* &* &1
 \end{array}\right)
\end{eqnarray*}
Then an optimal $(5,3,2,2)$ PDA can be obtained by deleting the last column of $\mathbf{P}_{3\times 6}$.
\begin{eqnarray*}
\mathbf{P}_{3\times 5}=\left(\begin{array}{ccc|ccc}
 0     &* &* &1      &* \\
 *&0      &* &* &1     \\
 *&* &0      &* &*
 \end{array}\right)
\end{eqnarray*}
\end{example}

\begin{theorem}\label{co-Z=1}
There exists an optimal $(K,F,1,m{F\choose 2}-\frac{(F-\kappa)(F-\kappa-1)}{2})$ PDA for any positive integers $K$ and $F$
satisfying  $K=(m-1)F+\kappa$ with $m\ge 1$ and $0\le \kappa<F$.
\end{theorem}

\begin{proof}
Let $\mathbf{P}_0$ be the original optimal $(F,F,1,{F\choose 2})$ PDA given in Theorem \ref{th-Ali-optimal} in place of  $k=F$ and $t=1$.
Note from  \eqref{Eqn_Def_AN_T}  that the last $F-\kappa$ columns in $\mathbf{P}_0$ contain exactly $\frac{(F-\kappa)(F-\kappa-1)}{2}$ integers which occur twice, so does
$\mathbf{P}$ in \eqref{eq-first-recursive}. That is, we can get an $(K,F,F-1,S')$ PDA $\mathbf{P}'$ by deleting the last $F-\kappa$ columns of $\mathbf{P}$ where
\begin{equation*}
S'=m{F\choose 2}-\frac{(F-\kappa)(F-\kappa-1)}{2}
\end{equation*}
Now we only need to show $\mathcal{S}(K,F,1)=S'$ by \eqref{eq-digui-bound} as follows
\begin{eqnarray*}\label{Z=1--minS}
\mathcal{S}(K,F,1)&\geq& \left\lceil\frac{(F-1)K}{F}\right\rceil+\left\lceil\frac{(F-2)}{F-1}\left\lceil\frac{(F-1)K}{F}\right\rceil\right\rceil+
\ldots+
\left\lceil\frac{1}{1+1}\left\lceil\frac{2}{1+2}\left\lceil\ldots\left\lceil\frac{(F-1)K}{F}\right\rceil \ldots\right\rceil\right\rceil\right\rceil\\
&=&((F-1)(m-1)+\kappa)+((F-2)(m-1)+\kappa)+\cdots+(\kappa(m-1)+\kappa)+((\kappa-1)(m-1)+\kappa-1)+\\
&&\cdots+(m-1+1)\\
&=&\frac{F(F-1)}{2}(m-1)+\kappa(F-\kappa)+\frac{\kappa(\kappa-1)}{2}\\
&=&m{F\choose 2}-\frac{(F-\kappa)(F-\kappa-1)}{2}
\end{eqnarray*}
where in the first identity we recursively use
\begin{eqnarray*}
\left\lceil\frac{(F-i-1)}{F-i}\left\lceil\frac{(F-i)}{F-i+1}\cdots \left\lceil\frac{(F-1)K}{F}\right\rceil\right\rceil\right\rceil=\left\{\begin{array}{ll}
(F-i-1)(m-1)+\kappa,& 0\le i< F-\kappa\\
(F-i-1)(m-1)+F-i-1, & F-\kappa\le i<F-1
\end{array}
\right.
\end{eqnarray*}
\end{proof}

\begin{example}\rm
\label{ex-Z=1}
When $F=4$ and $K=7$,
we can first obtain the following  optimal $(4,4,1,6)$ PDA in Theorem \ref{th-Ali-optimal} in place of  $k=4$ and $t=1$.
\begin{eqnarray*}
\mathbf{P}_{4\times 4}=\left(\begin{array}{cccc}
 0    &1     &  2  & * \\
 3    &4     &* &2     \\
 5    &*&4      &1      \\
*&5     &3      &0
\end{array}\right)
\end{eqnarray*}
Then an optimal $(7,4,1,12)$ PDA can be obtained by Theorem \ref{co-Z=1}.
\begin{eqnarray*}
\mathbf{P}_{4\times 7}=\left(\begin{array}{cccc|cccc}
 0    &1     &  2    &* & 6     & 7 & 8\\
 3    &4     &* &2      & 9     & 10 & *\\
 5    &*&4      &1      & 11      &* & 10\\
*&5     &3      &0      &* & 11 & 9
\end{array}\right)
\end{eqnarray*}
\end{example}

\subsection{Deleting method in the second recursive construction}

\begin{lemma}\label{co-base optimal PDA}
Suppose that there exists an optimal $(K_0,F_0,Z_0,\frac{(F_0-Z_0)K_0}{F_0}+F_0-Z_0-1)$ PDA with $F_0|(F_0-Z_0)K_0$.
Then there exists an optimal $(mK_0-\kappa,mF_0,(m-1)F_0+Z_0,\frac{(F_0-Z_0)K_0}{F_0}+F_0-Z_0-1)$ PDA  for integers $m\ge 2$ and $0\le \kappa<\frac{mF_0}{F_0-Z_0}$.
\end{lemma}

\begin{proof}
Let $\mathbf{P}$ be the optimal $(mK_0,mF_0,(m-1)F_0+Z_0,\frac{(F_0-Z_0)K_0}{F_0}+F_0-Z_0-1)$ PDA obtained by Construction  \ref{con-second recursive} based on the same optimal $(K_0,F_0,Z_0,\frac{(F_0-Z_0)K_0}{F_0}+F_0-Z_0-1)$ PDA.
Given a nonnegative integer $\kappa$, delete the last $\kappa$ columns of $\mathbf{P}$ to form an $(mK_0-\kappa,mF_0,(m-1)F_0+Z_0,\frac{(F_0-Z_0)K_0}{F_0}+F_0-Z_0-1)$ PDA $\mathbf{P}'$.
Note that $\left\lceil\frac{(F_0-Z_0)(mK_0-\kappa)}{mF_0}\right\rceil=\frac{(F_0-Z_0)K_0}{F_0}$ when $\kappa<\frac{mF_0}{F_0-Z_0}$.
Then, it follows from Corollary \ref{cor-case-1} that
\begin{eqnarray*}
\mathcal{S}(mK_0-\kappa,F_0,(m-1)F_0+Z_0)\geq \frac{(F_0-Z_0)K_0}{F_0}+(F_0-Z_0)-1
\end{eqnarray*}
Therefore, $\mathbf{P}'$ is optimal according to \eqref{Eqn_Def_S}.
\end{proof}

Combining Theorem \ref{co-Z=0} and Lemma \ref{co-base optimal PDA}, we get two optimal PDAs.
\begin{theorem} \label{co-t=k-2-OPDA}
For any positive integers $k$ and $m\ge 2$, there exists an optimal $(mk-\kappa,m{k \choose k-2},m{k \choose k-2}-k+1,k)$ PDA and an optimal $(m{k\choose 2}-\kappa,mk,mk-2,k)$ PDA for each integer $0\le \kappa<\frac{mk}{2}$.
\end{theorem}

\begin{example}\label{ex-optimal (4,6,3,4)PDA}
Based on an optimal $(4,6,3,4)$ PDA listed in Example \ref{E-pda}, there exists an optimal $(4m-x,6m,6m-3,4)$ PDA for any positive integers $m$ and $x<2m$ by Theorem \ref{co-t=k-2-OPDA}. When $m=3$ we have an optimal $(12,18,15,4)$ PDA  $\mathbf{P}$. And we can also obtain an optimal PDA by deleting the last $1$, $2$, $3$, $4$ and $5$ columns of $\mathbf{P}$ respectively.
\end{example}

Finally, by deleting some columns of the optimal PDAs in Theorem \ref{co-Z=0} and Corollary \ref{co-comb}, we are able to get optimal PDAs with more choices of the parameters $K$ and $F$.
\begin{theorem}\label{th-F=Z-2}
Suppose that $F$ and $K$ are positive integers with $F|K$.
Then there exists an optimal $(K, F, F-2, S)$ PDA  for any integer $F^3\geq 2K^2$,
with $S=\frac{2K}{F}+1$ if $(\frac{2K}{F}+1)|F$ and $S=\frac{2K}{F}+2$ otherwise.
\end{theorem}

\begin{proof}
Let $n = \frac{K}{F}$.
Since $F^3\geq 2K^2$ implies $F\geq 2n^2$, we can write $F=2an +b$, where $0\leq b < 2n$ and $a \ge n$,
i.e., $b < 2a$. Then, we only need to show that  there exists an optimal $(n(2an+b), 2an+b,2an+b-2,S)$ PDA  for any integer $b \in [0,2a)$,
with $S=2n+1$ if $b=a$ and $S=2n+2$ otherwise, which is divided  into three cases.
\begin{itemize}
  \item[(i)] When $b =a$, the optimal $(n(2an+a), 2an+a,2an+a-2, 2n+1)$ PDA can be obtained for any positive integers $a$ and $n$
   from Theorem \ref{co-Z=0}-4) in place of $m=a$ and $k=2n+1$.
  \item[(ii)] When $b \in [0,a)$, let $\mathbf{P}'$ be the optimal $(2n^2, 2n,2n-2,2n+2)$ in Corollary \ref{co-comb}.
  Let $\mathbf{P}''$ be the optimal $(n(2n+1), 2n+1,2n-1,2n+1)$ PDA in
  Theorem \ref{th-Ali-optimal} in place of $k=2n+1$ and $t=2n-1$. In Construction \ref{con-second recursive},
 let $\mathbf{P}_i = \mathbf{P}'$, $0 \leq i \leq a-b-1$, and $\mathbf{P}_{i}= \mathbf{P}''$, $a-b \leq i\leq a-1$,
 we can construct an $(n(2an+b), 2an+b,2an+b-2,2n+2)$ PDA which is optimal by Corollary \ref{co-tighter S(K,F,F-2)}.

  \item[(iii)] When $b \in (a,2a)$,  let $\mathbf{P}'''$ be the $(n(2n+2), 2n+2,2n,2n+2)$ PDA generated
  by deleting $n+1$ columns from  the optimal $((n+1)(2n+1), 2n+2,2n,2n+2)$ PDA
  in Theorem \ref{th-Ali-optimal} in place of $k=2n+2$ and $t=2n$, which is optimal by Corollary \ref{co-tighter S(K,F,F-2)}.
  In  Construction \ref{con-second recursive},
 set $\mathbf{P}_i = \mathbf{P}'''$, $0 \leq i \leq b-a-1$, and $\mathbf{P}_{i}= \mathbf{P}''$, $b-a \leq i \leq a-1$,
we can generate an $n(2an+b), 2an+b,2an+b-2,2n+2)$ PDA which is optimal by Corollary \ref{co-tighter S(K,F,F-2)}.
\end{itemize}

\end{proof}

\begin{example}\label{ex-h&(h+1)&(h+2)-F=Z-2}
From Theorem \ref{th-F=Z-2}, there exist: optimal $(F,F,F-2,S)$ PDA  for any integer $F\geq 2$,
  where $S=3$ when $3|F$ or $S=4$ when $3\nmid F$;  optimal $(2F,F,F-2,S)$ PDA  for any integer $F\geq 8$,
  where $S=5$ when $5|F$ or $S=6$ when $5\nmid F$;  optimal $(3F,F,F-2,S)$ PDA  for any integer $F\geq 18$,
  where $S=7$ when $6|F$ or $S=8$ when $6\nmid F$.
\end{example}

From Example \ref{ex-h&(h+1)&(h+2)-F=Z-2}, it is not difficult to see  that Theorem \ref{th-F=Z-2} in fact gives optimal PDAs with $F|K$ if $F$ is proper large.

\section{Conclusion}\label{c3conclusion}
In this paper, lower bounds on PDA were derived. With respect to the new bounds,  some new proposed PDAs are able to
shown to optimal. Particularly, optimal PDAs with $Z=1$ and $F-1$ for any positive integers $K$ and $F$ were obtained, and
several infinite classes of optimal PDAs with $Z=F-3,F-2$ were constructed.

Further, it would be of particular interest if we could find a tighter lower bound on $\mathcal{S}(K,F,Z)$ for some other $Z$s and more optimal PDAs.
The readers are invited to join the adventure.

\end{document}